\RequirePackage{pdf14}
\documentclass[preprint,fleqn,5p,twocolumn]{elsarticle}

\usepackage{amsmath, amssymb, bbm, xspace}
\usepackage{epsfig}
\usepackage{longtable}
\usepackage{color}
\usepackage{mathrsfs}
\usepackage{comment}
\usepackage{ifthen}
\newboolean{showcomments}
\setboolean{showcomments}{true}
\usepackage{courier}
\usepackage{xcolor}
\usepackage{todonotes}
\usepackage[framemethod=tikz]{mdframed}
\usepackage{lineno}
\newmdenv[leftmargin=\dimexpr-0.4em, innerleftmargin=0.5em,
rightmargin=\dimexpr-0.4em, innerrightmargin=0.5em,
linewidth=2pt,linecolor=red, topline=false, bottomline=false,
innertopmargin=0pt,innerbottommargin=0pt,skipbelow=0pt,skipabove=0pt,%
]{notex}

\newenvironment{note}%
{\vskip\dimexpr\dp\strutbox-\prevdepth\relax\notex\strut\ignorespaces}%
{\xdef\notetpd{\the\prevdepth}\endnotex\vskip-\notetpd\relax}

\let\oldtodo\todo

\makeatletter%
\DeclareDocumentCommand{\todo}{ O{} +g +d<> }{%
		\setlength{\marginparwidth}{1.5cm}%
	\IfNoValueTF{#2}{\relax}{%
		\oldtodo[caption={#2},size=\scriptsize,#1]{\renewcommand{\baselinestretch}{0.8}\selectfont\sffamily#2\par}%
	}%
	\IfNoValueTF{#3}{\relax}{%
		\IfNoValueTF{#2}{
			\begin{note}%
				\begin{internallinenumbers}%
					\indent%
					#3%
				\end{internallinenumbers}%
			\end{note}%
		}{
			\vspace{-0\baselineskip}%
			\begin{note}%
				\begin{internallinenumbers}%
					\indent%
					#3%
				\end{internallinenumbers}%
			\end{note}%
		}%
	}%
}%
\makeatother
\usepackage{soul}

\newcommand{\hlc}[2][yellow]{{%
		\colorlet{foo}{#1}%
		\sethlcolor{foo}\hl{#2}}%
}

\newcommand{\removetodo}[2]{\todo[color=pink]{\textbf{delete:} ``#1'' #2}\hlc[pink]{#1}}
\newcommand{\inserttodo}[1]{\todo[color=green!40]{\textbf{insert:} #1}}

\newcommand{\hltodoy}[2]{\todo[color=yellow!40]{#2}\hl{#1} }

\newcommand{\hltodoc}[3]{\todo[color=#3!40]{#2}\hlc[#3]{#1} }

\newcommand{\hltodo}[2]{\todo[color=orange!40]{#2}\hlc[orange!40]{#1} }
\newcommand{\replacetodo}[2]{\todo[color=pink!40]{\textbf{replace with:}``#2'' }\hl{#1} }

\usepackage{marginnote}
\newcommand{\todol}[1]{{%
		\let\marginpar\marginnote
		\reversemarginpar
		\renewcommand{\baselinestretch}{0.8}%
		\todo{#1}}}
	
\newcommand{\inserttodol}[1]{{%
		\let\marginpar\marginnote
		\reversemarginpar
		\renewcommand{\baselinestretch}{0.8}%
		\inserttodo{#1}}}
	
\newcommand{\removetodol}[2]{{%
		\let\marginpar\marginnote
		\reversemarginpar
		\renewcommand{\baselinestretch}{0.8}%
		\removetodo{#1}{#2}}}

\newcommand{\hltodol}[2]{{%
		\let\marginpar\marginnote
		\reversemarginpar
		\renewcommand{\baselinestretch}{0.8}%
		\hltodo{#1}{#2}}}

\newcommand{\replacetodol}[2]{{%
		\let\marginpar\marginnote
		\reversemarginpar
		\renewcommand{\baselinestretch}{0.8}%
		\replacetodo{#1}{#2}}}

\newcommand{\hltodoyl}[2]{{%
		\let\marginpar\marginnote
		\reversemarginpar
		\renewcommand{\baselinestretch}{0.8}%
		\hltodoy{#1}{#2}}}

\newcommand{\hltodocl}[3]{{		\let\marginpar\marginnote
		\reversemarginpar
		\renewcommand{\baselinestretch}{0.8}%
		\hltodoc{#1}{#2}{#3}}}

\newtheorem{theorem}{Theorem}[section]

\newtheorem{lemma}[theorem]{Lemma}

\newtheorem{proposition}[theorem]{Proposition}

\newtheorem{corollary}[theorem]{Corollary}

\newtheorem{definition}{Definition}[section]

\def\bkE{{\rm I\kern-.17em E}}
\def\bk1{{\rm 1\kern-.17em l}}
\def\bkD{{\rm I\kern-.17em D}}
\def\bkR{{\rm I\kern-.17em R}}
\def\bkP{{\rm I\kern-.17em P}}

\def\bkZ{{\bf{Z}}}

\def\bkE{{\rm I\kern-.17em E}}
\def\bk1{{\rm 1\kern-.17em l}}
\def\bkD{{\rm I\kern-.17em D}}
\def\bkR{{\rm I\kern-.17em R}}
\def\bkP{{\rm I\kern-.17em P}}

\makeatletter
\newcommand{\pushright}[1]{\ifmeasuring@#1\else\omit\hfill$\displaystyle#1$\fi\ignorespaces}
\newcommand{\pushleft}[1]{\ifmeasuring@#1\else\omit$\displaystyle#1$\hfill\fi\ignorespaces}
\makeatother

\def\bkZ{{\bf{Z}}}
\def\b12{(\beta_1,\beta_2)}

\newenvironment{example}{{\noindent \bf Example}}{\hfill $\square$\hspace{-4.5pt}\vspace{6pt}}
\newcounter{example}
\renewcommand{\theexample}{\thesection.\arabic{example}}

\newcounter{remark}
\renewcommand{\theremark}{\thesection.\arabic{remark}}

\def\Xscr{\mathcal{X}}
\def\Yscr{\mathcal{Y}}

\newlength{\noteWidth}
\long\def\notes#1{\ifinner
{\tiny #1}
\else
\marginpar{\parbox[t]{\noteWidth}{\raggedright\tiny #1}}
\fi\typeout{#1}}

 \def\notes#1{\typeout{read notes: #1}} 

\newcommand{\ut}{\mathscr{U}}

\newcommand{\ie}{i.e.\@\xspace} 

\def\OPT{{\rm OPT}}

\def\Pbb{{\mathbb{P}}}
\def\Nbb{{\mathbb{N}}}

\def\Ibf{{\bf I}}

\def\spose#1{\hbox to 0pt{#1\hss}}

\def\text #1{\hbox{\quad#1\quad}}

\def\xhat{{\hat x}}

\def\nthinsp{\mskip -2   mu}

\def\superstar{^{\raise 0.5pt\hbox{$\nthinsp *$}}}
\def\SUPERSTAR{^{\raise 0.5pt\hbox{$*$}}}

\def\lamstarT {\lambda^{\raise 0.5pt\hbox{$\nthinsp *$}T}}

\def\Iscr{{\cal I}}

\def\Mscr{{\cal M}}

\def\Pscr{{\cal P}}

\def\Wscr{{\cal W}}
\def\Mscr{{\cal M}}

\def\Rscr{{\cal R}}
\def\Gscr{{\cal G}}
\def\Cscr{{\cal C}}

\def\Xscr{{\cal X}}
\def\Yscr{{\cal Y}}

\def\xhat{\skew{2.8}\widehat x}

\def\Xhat{\widehat{\mathcal{X}}}

\def\supp{{\rm supp}}

\let\forallnew\forall
\renewcommand{\forall}{\forallnew\ }
\let\forall\forallnew

		\def\bkE{{\rm I\kern-.17em E}}
		\def\bk1{{\rm 1\kern-.17em l}}
		\def\bkD{{\rm I\kern-.17em D}}
		\def\bkR{{\rm I\kern-.17em R}}
		\def\bkP{{\rm I\kern-.17em P}}
		\def\bkY{{\bf \kern-.17em Y}}
		\def\bkZ{{\bf \kern-.17em Z}}
		\def\bkC{{\bf  \kern-.17em C}}

{\begin{list}{}%
         {\setlength{\leftmargin}{#1}}%
         \item[]%
}
{\end{list}}

		\def\bsp{\begin{split}}
		\def\beq{\begin{eqnarray}}
		\def\bal{\begin{align*}}
		\def\bc{\begin{center}}
		\def\be{\begin{enumerate}}
		\def\bi{\begin{itemize}}
		\def\bs{\begin{small}}
		\def\bS{\begin{slide}}
		\def\ec{\end{center}}
		\def\ee{\end{enumerate}}
		\def\ei{\end{itemize}}
		\def\es{\end{small}}
		\def\eS{\end{slide}}
		\def\eeq{\end{eqnarray}}
		\def\eal{\end{align*}}
		\def\esp{\end{split}}
		\def\qed{ \vrule height7.5pt width7.5pt depth0pt}  

	\def\cp2problem#1#2#3#4{\fbox
		 {\begin{tabular*}{0.9\textwidth}
			{@{}l@{\extracolsep{\fill}}l@{\extracolsep{6pt}}l@{\extracolsep{\fill}}c@{}}
				#1 & & $#4 $ 
			\end{tabular*}}}

		\def\bkE{{\rm I\kern-.17em E}}
		\def\bk1{{\rm 1\kern-.17em l}}
		\def\bkD{{\rm I\kern-.17em D}}
		\def\bkR{{\rm I\kern-.17em R}}
		\def\bkP{{\rm I\kern-.17em P}}
		
		\def\bkZ{{\bf{Z}}}

\newcommand {\beeq}[1]{\begin{equation}\label{#1}}
\newcommand {\eeeq}{\end{equation}}
\newcommand {\bea}{\begin{eqnarray}}
\newcommand {\eea}{\end{eqnarray}}

\def\texitem#1{\par\smallskip\noindent\hangindent 25pt
               \hbox to 25pt {\hss #1 ~}\ignorespaces}

\def\bsp{\begin{split}}
		\def\beq{\begin{eqnarray}}
		\def\bal{\begin{align*}}
		\def\bc{\begin{center}}
		\def\be{\begin{enumerate}}
		\def\bi{\begin{itemize}}
		\def\bs{\begin{small}}
		\def\bS{\begin{slide}}
		\def\ec{\end{center}}
		\def\ee{\end{enumerate}}
		\def\ei{\end{itemize}}
		\def\es{\end{small}}
		\def\eS{\end{slide}}
		\def\eeq{\end{eqnarray}}
		\def\eal{\end{align*}}
		\def\esp{\end{split}}
		\def\qed{ \vrule height7.5pt width7.5pt depth0pt}  

\usepackage{amsmath, amssymb, xspace}
\usepackage{epsfig}
\usepackage{longtable}
\usepackage{color}
\usepackage{mathrsfs}
\usepackage{subfig}
\newenvironment{proof}[1][]{{\noindent \emph {Proof} #1: }}{\hfill \qed \vspace{3pt}\\ }

\def\Cscr{{\cal C}}

\usepackage{blkarray}
\usepackage{amsmath, amsfonts, amssymb, float, enumerate, bm, graphicx, mathtools, mathrsfs}

\usepackage{enumitem}
\usepackage{tikz}

\usepackage{lipsum}
\newlength\mylen
\newlist{mycases}{enumerate}{1}
\setlist[mycases,1]{label=\textbf{Case~\arabic*.},
  labelwidth=\dimexpr-\mylen-\labelsep\relax,leftmargin=0pt,align=right}

\newlist{myclass}{enumerate}{1}
\setlist[myclass,1]{label=\textbf{Class~\arabic*.},
  labelwidth=\dimexpr-\mylen-\labelsep\relax,leftmargin=0pt,align=right}
\usetikzlibrary{arrows}
\usepackage{verbatim}
\usepackage{url}
\tikzstyle{int}=[draw, fill=white!20, minimum size=2em]
\tikzstyle{init} = [pin edge={to-,thin,black}]
\usepackage{tkz-euclide}
\usetikzlibrary{backgrounds}
\usetikzlibrary{shapes.geometric}
\restylefloat{table}

\newcounter{eg}[section]
\renewcommand{\theeg}{\arabic{section}.\arabic{eg}}
\newenvironment{examp}[1][]{\refstepcounter{eg}
   \textit{Example~\theeg. #1} \rmfamily}{\hfill $\square$   \hspace{-4.5pt} \vspace{6pt}}

\colorlet{red}{black}
\newcommand{\AS}{{\mathcal A}_S }
\newcommand{\AR}{{\mathcal A}_R }
\newcommand{\Ubr}{{\mathcal A}_S^*}

\newcommand{\brr}{\mathscr B}
\newcommand{\Rbb}{\mathbb R}

\newcommand{\Ninfo}{\mathfrak{I}}

\newcommand{\Ut}{\mathscr{U}}

\newcommand{\lpu}{\mathbf{P}(\ut)}
\newcommand{\dpu}{\mathbf{D}(\ut)}

\newcommand{\pie}{\pi_{k}}

\newcommand{\vpie}{\pi_{\varepsilon}}

\newcommand{\ese}{$\varepsilon$-SES}

\makeatletter
\newcommand*{\rom}[1]{\expandafter\@slowromancap\romannumeral #1@}
\makeatother

\usepackage{algorithmic}
\usetikzlibrary{arrows}
\usepackage{textcomp}
\newcounter{casenum}

\def\BibTeX{{\rm B\kern-.05em{\sc i\kern-.025em b}\kern-.08em
    T\kern-.1667em\lower.7ex\hbox{E}\kern-.125emX}}

\usepackage{natbib}
\begin{document}
\title{\textbf{Informativeness and Trust in Bayesian Persuasion}
}

\author[1]{Reema Deori}
\author[1]{Ankur A. Kulkarni}

\address[1]{{Center for Systems and Control, Indian Institute of Technology Bombay},
{Powai},
{Mumbai},
	{400076},
	{Maharashtra},
	{India. deori.reema@iitb.ac.in, kulkarni.ankur@iitb.ac.in}}

\begin{abstract}
A persuasion policy successfully persuades an agent to pick a particular action only if the information is designed in a manner that convinces the agent that it is in their best interest to pick that action. Thus, it is natural to ask, what makes the agent trust the persuader's suggestion?
We study a Bayesian persuasion interaction between a sender and a receiver where the sender has access to private information about a source and the receiver attempts to recover this information from messages sent by the sender. The sender crafts these messages in an attempt to maximize its utility which depends on the source symbol and the symbol recovered by the receiver. Our goal is to characterize the min-max equilibrium utility of the sender, called the \textit{Stackelberg game value}, and the amount of true information revealed by the sender during persuasion.
We find that the Stackelberg game value is given by the optimal value of a \textit{linear program} on probability distributions constrained by certain \textit{trust constraints}. These constraints encode that any signal in a persuasion strategy must contain more truth than untruth and thus impose a fundamental bound on the extent of obfuscation a sender can perform in any attempt to persuade the receiver.
We define \textit{informativeness} of the sender as the minimum value of the expected number of symbols truthfully revealed by the sender in any accumulation point of a sequence of $\varepsilon$-equilibrium persuasion strategies, and show that it is given by the optimal value of another linear program. Informativeness gives a fundamental bound on the amount of information the sender must reveal to persuade a receiver. Closed form expressions for the Stackelberg game value and the informativeness are presented for structured utility functions. This work generalizes our previous work~\cite{deori2022information} where the sender and the receiver were constrained to play only deterministic strategies and a similar notion of informativeness was characterized. Comparisons between the previous and current notions are discussed.

\end{abstract}
\maketitle
\section{Introduction}
Consider a social-media influencer who makes money through endorsement offers from brands based on the number of purchases made using the links affiliated to her account. The products she reviews are of varying quality, whereby she faces a dilemma between maximizing sales and maintaining trust.  For example, when she is paid for the promotion of a poor quality product, honesty will lead to low sales. However, frequent misleading reviews will hurt her too, since this action would result in her followers losing trust in her honest reviews as well. A better choice in such scenarios appears to be to maintain some ambiguity: promote the poor quality product as an average quality product by giving average scores to both an average quality product and the poor quality product, thereby partially misleading her followers and also partially winning their trust. In other words, the influencer's optimal policy `lies'\footnote{Pun intended.} in striking a balance between maintaining trust through honesty and maximizing sales through prevarication.

Our aim is to understand how much truth there is in such an influencer's words. We model an interaction between a strategic informed sender and an uniformed receiver. The sender persuades the receiver by committing to a randomized signaling policy crafted with the aim of maximizing its utility. The receiver on the other hand attempts to know the true information of the sender.
Our contribution is a characterization of the sender's  expected utility in a Stackelberg equilibrium~\cite{basar1999dynamic} as a linear program. We also get a linear programming characterization of the \textit{minimum amount of truth} revealed to the receiver in any Stackelberg equilibrium. In the process we uncover a key element of such problems: persuasion must be subject to ``trust constraints'' for it to work. The influencer's optimal policy must reveal some truth.

We studied the same game previously in \cite{deori2022information} but under the restriction that both players play only deterministic strategies. Thus, the sender's objective in this game was to pick an `encoding' strategy that persuades the receiver to recover the majority of the symbols as the sender's preferred choice of symbols. Our study in \cite{deori2022information} showed that every equilibrium strategy is equivalent to a \textit{vertex clique cover} of a suitably defined \textit{strong sender graph}. We then characterized the \textit{informativeness of a sender}, \ie, the minimum amount of information recovered by the receiver in any equilibrium, and demonstrated that it is given by the \textit{vertex clique cover number} of the strong sender graph.
In this paper, we take this thought forward allowing the players to play randomized strategies where obfuscation takes a more subtle and interesting form.

\subsection{Main findings}
Recall the dilemma the influencer had between honesty and prevarication, or equivalently between retaining the trust of her followers and maximizing her revenue. Intuitively, to maintain the trust of its followers, the influencer must always reveal a greater degree of truth than untruth in every review.
Thus, any review, say a ``low'' rating,  is trustworthy only if it is mostly given to low quality products.
Our central contribution lies in formalizing this intuition.

We formulate the above setting as that of Bayesian persuasion and solve for a min-max Stackelberg equilibrium solution in behavioural strategies.
We find that even in simple problem classes, finding the Stackelberg equilibrium value (SGV) of the sender involves a long and complex calculation. Our first result bypasses this difficulty: it shows that the equilibrium expected utility of the sender is characterized by a \textit{linear program} (LP) on probability distributions; these distributions are required to meet certain \textit{trust constraints}. These constraints encode that any signal must contain more truth than untruth and thus impose a fundamental bound on the extent of obfuscation a sender can perform. These constraints make the receiver trust the sender about the truth of some symbols and accordingly makes the rational receiver choose a best response which recovers some symbols as symbols preferred by the sender. We measure the information revealed by the sender by the minimum value of the expected number of symbols truthfully revealed by the sender in any accumulation point of a sequence of $\varepsilon$-Stackelberg equilibrium strategies. This quantity serves as a fundamental bound on the amount of truth revealed in any equilibrium strategy and is a measure of the informativeness of this interaction for the receiver. We call this quantity the \textit{informativeness} of the sender and show that, despite a complex definition, it is also given by another linear program.  A corollary of our analysis is that the sender will opt for full information disclosure at equilibrium \textit{if and only if} there is full  alignment of objectives with the receiver.
Using this result, we prove that whenever the non-negative terms in the utility of the sender are positive and constant, the loss of information is identical in all persuasion strategies. For such utility functions, we prove that the SGV is a constant multiple of the amount of loss of information at equilibrium.

We end our paper by deriving closed form expressions for the SGV and informativeness for three different classes of utility functions. We introduce the \textit{obfuscation graph} of a utility function and characterize the SGV and informativeness exactly for utilities for which this graph is chain, cycle or star. We conclude by comparing the notion of informativeness with the notion of informativeness in~\cite{deori2022information} where the players are allowed to play only deterministic strategies.

Informativeness serves as a measure of information content in a utility function, much like entropy does for a probability distribution. It enjoys some natural properties and has a clean characterization. This work, following up on our previous work~\cite{deori2022information} and similar works in a screening setting~\cite{vora2024shannon,vora2023achievable} furthers our understanding information exchange in multiagent interactions.
\subsection{Related work and organization}
Sequential strategic interactions with incomplete information are covered mainly by \textit{screening games} and \textit{signalling games}~\cite{rasmusen2007games}. The order of play separates one regime from the other wherein the informed player makes the first move in signalling games, unlike the screening regime where the uninformed player leads. The Bayesian persuasion model introduced in~\cite{kamenica2011bayesian} is a popular approach to model persuasion and study \textit{information design} in the signalling regime. We also use this framework in this paper. Bayesian persuasion literature has grown tremendously in the recent past and issues of information, which is our interest, have come to the fore. For example \cite{le2019persuasion} studies persuasion with communication constraints, \cite{rouphael2021strategic} explores a multi-user Bayesian persuasion setting in an information-theoretic framework and \cite{akyol2015strategic} studies strategic information transfer between a transmitter and a receiver in a signalling set-up. Previously, we also studied a signalling set-up in \cite{deori2022information,deori2023zero} but with deterministic strategy space. The screening version of this interaction was studied in \cite{vora2024shannon}. However, a formal understanding of the information exchange in Bayesian persuasion has remained open. We seek to fill this gap in this paper.

The paper is organized as follows. We formulate the problem in Section \ref{sec2} and introduce the LP in Section \ref{sec3}, where we characterize the Stackelberg equilibrium. Section \ref{sec4} characterizes the informativeness while Section \ref{sec5} is dedicated to understanding how SGV and informativeness varies for different class of utility functions using graph theoretic interpretations. Section \ref{sec6} concludes the paper.

\section{Problem formulation}\label{sec2}
 \subsection{Notation}
We use $\OPT(\bullet)$ to denote the optimal value of the optimization problem `$\bullet$' and  $\Pscr(\bullet)$ to denote the set of probability distributions on `$\bullet$'.
 \subsection{Sender-receiver game}
We consider a source alphabet $\Xscr$ of size $q$; each element of $\Xscr$ is called a \textit{symbol}. Let $\Yscr$ be the set of possible \textit{signals} that the sender can assign to the elements of $\Xscr$. We assume $|\Xscr|=|\Yscr|$. The source generates a symbol $X$ uniformly at random with a probability $\Pbb(X = x) =
\frac{1}{q}, \forall x \in \Xscr$. The sender maps $X$ randomly to a signal $Y \in \Yscr$ according to some distribution $\pi \in \Pscr(\Xscr|\Yscr)$. The receiver attempts to recover $X$ from $Y$ using a distribution $\sigma \in \Pscr(\Yscr|\Xscr)$.
If the receiver recovers the source symbol $x\in \Xscr$ as $x'\in \Xscr$ then the sender obtains utility $\ut(x',x)$, where $\ut:\Xscr \times \Xscr \rightarrow \Rbb$. The sender attempts to maximize this utility by choosing $\pi$. Without loss of generality, we assume that $\ut(x,x)=0,\forall x \in \Xscr$.

Let $\AS:=\{\pi| \pi \in \Pscr(\Xscr|\Yscr) \}$ and $\AR:=\{\sigma| \sigma \in \Pscr(\Xscr|\Yscr) \}$ denote the collection of the strategies of the sender and the receiver respectively. For any $\pi \in \AS$, let $$\Yscr(\pi):=\{y\in \Yscr| \exists x \in \Xscr \text{s.t.} \pi(y|x)>0\}$$ be the set of possible signals used by the sender with positive probability when it plays $\pi$. Let $\widehat{X}$ denote the symbol recovered by the receiver. Thus, for a given pair of $\pi \in \AS$ and $\sigma\in \AR,$  the joint distribution of $ X,Y,\widehat{X}$ is given by
\begin{equation}
\Pbb_{\pi,\sigma}(x,y,\xhat) = \frac{1}{q}\pi(y|x)  \sigma(\xhat|y).
\end{equation}

The receiver's goal is to choose $\sigma\in \AR$ to maximize the probability of correct recovery of the symbols generated at the source. This is equivalent to maximizing
\[ \Rscr(\pi,\sigma) := q\Pbb_{\pi,\sigma}(X=\widehat{X})= \sum_{x\in \Xscr,y\in \Yscr} \pi(y|x)\sigma(x|y).\]

The set of strategies of the receiver which maximize $\Rscr(\pi,\sigma)$ is called the \textit{best response set}, $\brr(\pi),$
$$\brr(\pi)=\{\sigma \in \AR(\pi)|\ \sigma \in \arg \max \limits_{\sigma \in \AR(\pi)} \Rscr (\pi,\sigma)\}.$$
Clearly, $\brr(\pi)$ can be also be expressed as
\begin{equation}\label{eq:brr_eq}
    \brr(\pi)=\{\sigma \in \AR|\supp(\sigma(\bullet|y))\equiv \arg \max \limits_{x} \pi(y|x)\}.
\end{equation}

For a $\pi$ and $\sigma \in \brr(\pi), $ let \begin{equation}
\Xhat(\pi,\sigma):=\{\xhat\in \Xscr| \ \exists y \in \Yscr(\pi) \text{s.t.} \sigma(\xhat|y)>0 \},
\end{equation} be the set of symbols which have positive probability of getting recovered when the sender and the receiver play $\pi$ and $\sigma$ respectively. For a $\pi \in \AS$ and an $x\in \Xscr$, we define the support of $\pi(\cdot|x)$ as $$E_x(\pi):=\{y \in \Xscr : \pi(y|x)>0\}.$$ The sender's goal is to pick a $\pi $ which maximizes its expected utility, \ie, to maximize,
\begin{equation}
U(\pi, \sigma):=\sum \limits_{x \in \Xscr} \sum \limits_{ y \in \Yscr(\pi)} \sum \limits_{\hat{x} \in \Xscr}\pi(y|x)\sigma (\hat{x}|y)\ut(\hat{x},x).
\end{equation}
From this point forward, $U(\pi,\sigma)$ will be referred to as the \textit{expected utility} of the sender. For any strategy $\pi \in \AS$, let, \begin{equation}
\begin{split}
\underline{U} (\pi)=\min \limits_{\sigma \in \brr(\pi) } U(\pi,\sigma) \text{and}
\overline{U}(\pi)=\max \limits_{\sigma \in \brr(\pi) } U(\pi,\sigma)
\end{split}
\end{equation}
denote the corresponding \textit{worst case expected utility} (WCEU) and the  \textit{best case expected utility} (BCEU) obtained by the sender, respectively. Let \begin{equation}
 \underline{\brr}(\pi)=\{\sigma \in \brr(\pi)| U(\pi,\sigma)=\min \limits_{\sigma' \in \brr(\pi) } U(\pi,\sigma')\}
\end{equation}
denote the collection of all those best response strategies of the receiver which give the sender the worst case expected utility when the sender plays $\pi.$ Let \begin{equation}
D(\pi)= \{\sigma \in \brr(\pi)| \sigma(x|y) \in \{0,1\}, \forall x\in \Xscr, y \in \Yscr\}
\end{equation} be the collection of \textit{deterministic} best response strategies of the receiver and let \begin{equation}
 \overline{D}(\pi)=\{\sigma \in D(\pi)| U(\pi,\sigma)=\max \limits_{\sigma' \in \brr(\pi) } U(\pi,\sigma')\}.
 \end{equation} be the collection of all those deterministic best response strategies of the sender which give the sender the \textit{best case expected utility}.

 We study this sender-receiver interaction as Bayesian persuasion setting and seek a \textit{min-max Stackelberg equilibrium} where the sender commits first. Thus, $ \pi^* \in \arg \sup \limits_{\pi \in \AS} \underline{U} (\pi) $
is a \textit{Stackelberg equilibrium strategy} of the sender and
\begin{equation}\label{eq:firstOptimization}
\ut^*: =\sup \limits_{\pi \in \AS} \underline{U} (\pi),
\end{equation}
is the \textit{Stackelberg game value} (SGV).
We also introduce another quantity $\overline\ut^*$ \begin{equation}\label{eq:firstOptimization}
\overline\ut^*: =\sup \limits_{\pi \in \AS} \overline{U} (\pi)
\end{equation} to quantify the supremum of the best case expected utility of the sender.

In our study we find that there exists a sender strategy $\overline{\pi}^* \in \AS$ which attains $\overline{\ut}^*$, \ie  results in $\overline{U}(\overline{\pi}^*)=\overline{\ut}^*.$ 
Unfortunately, there seem to generically exist utility functions (see Example \ref{example1} below)  for which the supremum $\Ut^*$ is not attained by any $\pi \in \AS$. But remarkably, we find that $\ut^*=\overline{\ut}^*.$  Since a Stackelberg equilibrium need not exist, we base our analysis on the $\varepsilon$-\textit{Stackelberg equilibrium} defined below.
\begin{definition}($\varepsilon$-Stackelberg equilibrium strategy) Let $\varepsilon \geq 0$. A strategy $\pi^{*} \in \AS$ is a $\varepsilon$-Stackelberg equilibrium strategy (\ese) of the sender if
\begin{equation}\label{eq:S.E.1}
\Ut^* \geq  \underbar{U} (\pi ^*)   \geq \Ut^*-\varepsilon.
\end{equation}
If $\pi^*$ is a \ese\ of the sender then every $\sigma^* \in \brr(\pi^*)$ is a \ese \ of the receiver when the sender plays $\pi^*$.
		\end{definition}
If $\varepsilon=0$, then we shall call the pair of $\pi^*$ and $\sigma^* \in \brr(\pi^*)$ a pair of \textit{Stackelberg equilibrium strategies}. In the following section we present our analysis of the Stackelberg game value.

\section{LP formulation for the SGV}\label{sec3}
In this section we  prove that the SGV is given by the optimal value of a linear program.  For computing the SGV, one would ordinarily compute $\underline{U}(\pi)$ for every strategy in $ \pi \in \AS$, and thereafter optimize over $\pi$. The LP we present not only significantly simplifies this computation but also yields insight into truthful revelation by the sender. This result, shown in Theorem \ref{theo:main} below, is the main contribution of our paper.
\subsection{Trust constraints}
Before we present the actual LP formulation, we provide a bit of motivation. For any pair of $\pi$ and $\sigma\in \brr(\pi)$, and $x,\xhat\in \Xscr$ define $\mu(\xhat|x)$ as
\begin{equation}\label{eq:mu_def}
\small
\mu (\xhat|x)=\sum \limits_{y \in \Yscr} \pi(y|x)\sigma(\xhat|y).
\end{equation} It is easy to check that  $\mu(\xhat|x) = \Pbb_{\pi,\sigma}(\widehat{X}=\xhat|X=x)$.
We shall call $\mu$ \textit{equivalent} to a pair of $\pi\in \AS$ and  a $\sigma \in \brr(\pi)$ (denoted $\mu \equiv (\pi,\sigma)$), if \eqref{eq:mu_def} holds for all $x,\xhat\in \Xscr$.
With a slight abuse of notation, let \begin{equation}
\widehat{\mathcal{X}}(\mu)=\{x\in \Xscr| \quad \mu(x|x)>0\}.
\end{equation} Notice that if $\mu \equiv (\pi,\sigma)$, then $\widehat{\mathcal{X}}(\mu)=\widehat{\mathcal{X}}(\pi,\sigma)$. Thus, if \begin{equation}\label{eq:opt_strat_struc}
V(\mu):= \sum \limits_{x,\xhat \in \Xscr}\mu(\xhat|x) \ut(\xhat,x),
\end{equation} then $U(\pi,\sigma)=V(\mu).$ Hence, the sender's objective is to identify a $\pi \in \AS$ such that the $\mu$ constructed using $\pi$ and  a $\sigma \in \underline{\brr}(\pi)$ must give $V(\mu)=\Ut^*$. Our goal is to eliminate $\pi,\sigma$ from this description and state the sender's objective directly in terms of $\mu.$
The challenge of course is that while any pair of $\pi$ and $\sigma \in \brr(\pi)$ corresponds to a  $\mu \in \Pscr(\Xscr|\Xscr)$, not every $\mu \in \Pscr(\Xscr|\Xscr)$ corresponds to a pair of $\pi$ and $\sigma \in \brr(\pi)$.

To constrain the allowable $\mu$'s further, observe the following. Every $\mu$ constructed using a pair of $\pi$ and $\sigma \in \brr(\sigma)$ must satisfy \begin{equation}\label{eq:trustdef}
  \mu(\xhat|\xhat)\geq \mu(\xhat|x),\forall x,\xhat \in \Xscr.
  \end{equation}
In other words, the probability of recovering $\xhat \in \Xscr$ as itself must be no less than that of recovering any $x$ as $\xhat$. The persuasion of the sender induces a probability distribution of recovery that is \textit{greater} for correct recovery than it is for incorrect recovery. We call these the \textit{trust constraints}. A persuasion strategy of the sender works for a receiver who wants to recover the truth \textit{only if} it wins the receiver's trust by obeying the trust constraints.

Remarkably, we find that the \textit{converse} is also true: every distribution $\mu$ which satisfies the trust constraints can be constructed from some pair of $\pi$ and $\sigma \in \brr(\pi)$. This paves the way for a simple, linear programming based characterization of the SGV. We prove this in Theorem \ref{theo:main}.

We first prove the validity of \eqref{eq:trustdef}
in the following lemma.

\begin{lemma}Every $\mu$ constructed from a pair of $\pi$ and a $\sigma \in \brr(\pi)$ using \eqref{eq:mu_def} satisfies the trust constraints.
\end{lemma} \label{lemma:mu_pi_equiv}
\begin{proof} Fix two distinct symbols $x$ and $\xhat \in \Xscr$. From \eqref{eq:brr_eq}, it is clear that if $\sigma(\xhat|y)>0$ then $\pi(y|\xhat)\geq \pi(y|x),\forall y \in \Yscr$.  Multiplying both sides by $\sigma(\xhat|y)$ and summing over $y\in \Yscr$ gives
 $\sum \limits_{y \in \Yscr} \pi(y|\xhat)\sigma(\xhat|y)
  \geq \sum \limits_{y \in \Yscr} \pi(y|x)\sigma(\xhat|y)$,  which results in
  $\mu(\xhat|\xhat) \geq \mu(\xhat|x).$ This proves our lemma.
  \end{proof}
Therefore, an expected utility of $V(\mu)$ is attainable by a sender only when $\mu$ also satisfies the trust constraints.

We now show another lemma that gives the sender the BCEU.
\begin{lemma}\label{Lemma:determisiticBR} For every $\pi$, $\overline{D}(\pi)\neq \phi$.
\end{lemma}
\begin{proof}
Notice that $U(\pi,\sigma)$ is linear in $\sigma$ for a fixed $\pi$. Thus, for any fixed $\pi$, finding a $\sigma\in \brr(\pi)$ which maximizes $U(\pi,\sigma)$ is equivalent to finding a $\sigma(\bullet|y)$ which maximizes $\Bigl(\sum \limits_{x\in \Xscr}\pi(y|x)\ut(\hat{x},x)\Bigl ) \sigma (\hat{x}|y)$ which is linear program for every $y\in \Yscr(\pi)$. Using \eqref{eq:brr_eq}, it follows that there exists a deterministic optimal  $\sigma\in \brr(\pi)$. This proves our theorem.
\end{proof}
Let \begin{equation}
\Ubr:=\{\pi \in \AS||\brr(\pi)|=1 \},
\end{equation}
be the collection of sender strategies which have a unique best response. It is easy to see that if $\pi \in \Ubr$, then the unique $\sigma \in \brr(\pi)$ must be deterministic (using Lemma \ref{Lemma:determisiticBR}).

\subsection{A linear program with trust constraints}
    Consider the LP denoted by $\lpu$, where 
    \begin{alignat}{4}
        \lpu:\quad \max_{\mu}& & \sum \limits_{x,\xhat}\mu(\xhat|x)\ut(\xhat,x) & & \nonumber             \\
        \text{s.t.} & &  \mu\in\Pscr(\Xscr|\Xscr),\\
        & &  \mu \text{satisfies } \eqref{eq:trustdef}. \label{eq:lpc1}
    \end{alignat}
In this section we prove our main result,  where we show that the optimal value of $\lpu$ is exactly equal to the SGV. We begin by first proving that in every optimal solution $\mu^*$, positive probability is never assigned to negative utility values. We prove this using the dual of $\lpu$ denoted by $\dpu$, where
\begin{equation}
\begin{aligned}
\dpu:\min_{w,v} \quad & \sum \limits_{x \in \Xscr} w(x),\\
\textrm{s.t.} \quad & w(x)- \sum \limits_{\xhat\neq x \in \Xscr } v(x,\xhat)\geq 0, \forall x\in \Xscr \\
\quad & w(x) + v(\xhat,x)-\ut(\xhat,x)\geq 0,\forall x \neq \xhat \in \Xscr\\
\quad & v(\xhat,x)\geq 0,\forall x \neq \xhat \in \Xscr\\
\quad & w(x), \text{unrestricted}, \forall x \in \Xscr.
\end{aligned}
\end{equation}
\begin{proposition}\label{prop:positive_prob_positive_values}
Let $\mu^*$ be an optimal solution of $\lpu$. If for some $x$ and $\xhat \in \Xscr$, $\ut(\xhat,x)<0$, then $\mu^*(\xhat|x)=0.$
\end{proposition}
\begin{proof} Fix a pair of $x$ and $\xhat \in \Xscr$ for which $\ut(\xhat,x)<0.$ Suppose that there exists an optimal solution $\mu^*$ such that $\mu^*(\xhat|x)>0$. Then by using the complementary slackness condition of $\lpu$, we have
$\ut(\xhat,x) = v(\xhat,x)+w(x).$
Recall that in $\dpu$, we had $w(x)\geq  \sum \limits_{\xhat} v(x,\xhat) \implies w(x)\geq 0$. Therefore, $\ut(\xhat,x)\geq 0$, which is a contradiction since $\ut(\xhat,x)<0.$ This proves our proposition.
\end{proof}
In the following lemma, we prove the existence of  an infinite sequence of strategies of the sender   whose WCEU converges to $\OPT(\lpu)$.
\begin{lemma}\label{lemma:sequence_convergent}
Given any utility function $\ut$, there exists a sequence of strategies $\{\pi_k\}_{k\in \Nbb}$ such that $\lim \limits_{}$ $\lim \limits_{k\rightarrow \infty} \underline{U}(\pi_k)=\OPT(\lpu)$.
\end{lemma}
\begin{proof}
Fix an optimal solution $\mu^*$ of $\lpu$.  We will first construct a $\pi^* \in \AS$  and $\sigma^*\in \brr(\pi^*)$ such that $\mu^*\equiv(\pi^*,\sigma^*)$. To this end pick distinct elements $y_i\in \Yscr$ for each $i\in \Xhat(\mu^*)$. Define $\pi^*$ and $\sigma^*$ as follows:
\begin{equation}\label{eq:mu_pi_def}
\begin{split}
\pi^*(y_i|x)  &=\mu^*(i|x), \qquad \forall i \in \Xhat(\mu^*),\forall x \in \Xscr \text{and}\\
\sigma^*(i|y_i) &=1, \forall i \in \Xhat(\mu^*).
\end{split}
 \end{equation}
Clearly $\sigma^*$ is deterministic and $\Yscr(\pi^*)=\{ y_i|i \in \Xhat(\mu^*)\}$ which results in $|\Yscr(\pi^*)|=|\Xhat(\mu^*)|$. Note that $\mu^*$ satisfies the trust constraints. As a consequence,
\begin{equation}\label{eq:ineq_equality}
\pi^*(y_{i}|i)\geq \pi^*(y_i|x), \forall x \in \Xscr,\forall i \in \Xhat(\mu^*)
\end{equation} which makes $\sigma^*\in D(\pi^*)$.
Since $\mu^*(\xhat|x)\equiv\sum \limits_{y\in \Yscr}\pi^*(y|x)\sigma^*(\xhat|y)$ and $V(\mu^*)=U(\pi^*,\sigma^*)$, we get $\mu^*\equiv (\pi^*,\sigma^*)$.

 Next we show that $\sigma^*\in \overline{D}(\pi^*)$. Notice that if $\sigma^* \notin \overline{D}(\pi^*)$, then there exists another $\sigma'\in \overline{D}(\pi^*)$ such that $U(\pi^*,\sigma^*)<U(\pi^*,\sigma')$. In such scenario, we can construct another $\mu'$ from the pair of $\pi^*$ and $\sigma'$ such that $V(\mu^*)<V(\mu')$. But $\mu'$ is feasible for $\lpu$ which is a contradiction to the assumption that $\mu^*$ is optimal. Therefore,
$\sigma^* \in \overline{D}(\pi^*)$. Now $\pi^*$ can be categorized into two different classes based on the relationship between $\underline{U}(\pi^*)$ and $V(\mu^*)$.
\begin{mycases}
\item $\underline{U}(\pi^*)=V(\mu^*)$: For any $\pi^*$ of this class, the result follows trivially.
\item $\underline{U}(\pi^*)<V(\mu^*)$: For every $\pi^*$ of this class there exists a distinct $\sigma'\in \brr(\pi^*)$  such that $U(\pi^*,\sigma')<U(\pi^*,\sigma^*)=V(\mu^*)$. Our main goal is to come up with a sequence $\{\pie\}_{k\rightarrow \infty}$ such that $\lim \limits_{k\rightarrow \infty} \underline{U}(\pi_k)=\OPT(\lpu)$. To ensure this we want $\underline{U}(\pie)$ to approach $U(\pi^*,\sigma^*)$ as $k\rightarrow \infty$. We construct our $\pie$'s in such a way that $\forall k \in \Nbb,$ the corresponding $\pie \in \Ubr.$ This give us control over the value of $\underline{U}(\pie)$ for every $k$ since there exists only one best response strategy. We also want the construction to ensure that $k\rightarrow \infty$, $(\OPT(\lpu)-\underline{U}(\pie))\rightarrow 0.$ This will prove our theorem. To proceed with the construction, we introduce the following two sets. For a $i \in \Xhat(\mu^*)$ and $x\in \Xscr$, let \begin{equation*}
\begin{split}
Q(i)&=\{x  \in \Xscr| x \neq i, \pi^*(y_i |i)=\pi^*(y_i|x)\},\\
Z(x)&=\{i \in \Xhat(\mu^*)|x\in Q(i)\}.
\end{split}
 \end{equation*}
$Q(i)$ is the collection of all those symbols  $x\in \Xscr$ for which equality holds in \eqref{eq:ineq_equality}. Hence, $Q(i)$ comprises of all $x\in \Xscr$ distinct from $i$  that the receiver can map $y_i$ to in a best response.
$Z(x)$ is the collection of all $i\in \Xhat(\mu^*)$ for which equality holds in \eqref{eq:ineq_equality} for the fixed $x$. Hence, $Z(x)$ represents the collection of all signals distinct from $y_x$ that the receiver can map to $x$ as a best response. To construct such a $\pie$ from $\pi^*$, we need to define $\pie(\bullet|x)$ for all $x\in \Xscr.$ First we partition $\Xscr$ into two main classes based on the structure of $\mu^*$ and $\pi^*$.
\begin{myclass}
\item $x\in \Xhat(\mu^*)$: This class can be further subdivided into the following two classes:
\begin{enumerate}[label=(\alph*)]
\item $x\notin \cup_{i \in \Xhat(\mu^*)}Q(i)$ :
 For such $x$, define \[\pie(y_i|x)=\pi^*(y_i|x),\forall i \in \Xhat(\mu^*).\] Clearly,  $\pie(\bullet|x)$ is a probability distribution.
\item $x \in \cup_{i \in \Xhat(\mu^*)}Q(i)$: For such an $x$ of this class and an $y_i\in Z(x)$, we define \begin{equation}
\begin{split}
\pie(y_i|x)& =\pi^*(y_i|x)- \frac{\delta}{k},\\
\end{split}
\end{equation}
where $\delta>0$ is small.
To balance the weight of the distribution, we define:
\begin{equation}\label{eq:gh}
\begin{split}
\pie(y_x|x)&= \pi^*(y_x|x)+ |Z(x)|\frac{\delta}{k}\\
\pie(y_i|x)&= \pi^*(y_i|x),\forall i \in \Xhat(\mu^*)\backslash (Z(x) \cup\{x\}).
\end{split}
\end{equation}
This definition ensures that $\pie(\bullet|x)$ is a probability distribution.
\end{enumerate}
\item $x\notin \Xhat(\mu^*)$: This class can be further subdivided into two classes:
\begin{enumerate}[label=(\alph*)]
\item $x\notin \cup_{i \in \Xhat(\mu^*)}Q(i)$: 
For an $x$ of this class we define \[\pie(y_i|x)=\pi^*(y_i|x), \forall y_i \in \Yscr(\pi^*).\] 
Clearly, $\pie(\bullet|x)$ is a probability distribution.
\item $x \in \cup_{i \in \Xhat(\mu^*)}Q(i)$: For every $x$ in this class  let \begin{equation}\label{eq:c2b}
\begin{split}
\pie(y_i|x) &=\pi^*(y_i|x)- \frac{\delta}{k},\forall i\in Z(x)\\
\pie(y_i|x) &= \pi^*(y_i|x),\forall i \in \Xhat(\mu^*)\backslash Z(x).
\end{split}
\end{equation}
Next  for every $x$ belonging to this class consider a distinct $y_x\in \Yscr\backslash\Yscr(\pi^*)$. Let $\Yscr'(\pi^*)$ be the collection of all such $y_x$.  
Let \begin{equation}
\begin{split}
\pie(y_x|x)&= |Z(x)|\frac{\delta}{k},\forall y_x\in \Yscr'(\pi^*).
\end{split}
\end{equation}
\end{enumerate}
\end{myclass}
\end{mycases}
From the construction of $\pie$, it is evident that $\Yscr(\pie)=\Yscr(\pi^*)\cup \Yscr'(\pi^*)$, a constant independent of $k$ and $\pie\rightarrow \pi^*$.
Now we prove that this construction ensures \begin{equation}\label{eq:proving_unique}
\pie(y_i|i)>\pie(y_i|x),\forall x \in \Xscr\backslash\{i\},\forall y_i \in \Yscr(\pie).
\end{equation}
 Observe that \begin{align}
\pi^*(y_i|x)& \geq \pie(y_i|x),\forall i\neq x \label{eq:construc_conseq} \text{and}\\
\pie(y_x|x)& \geq \pi^*(y_x|x),\forall x\in \Xscr.\label{eq:construc_conseq2}
\end{align}
We will first categorize $y_i\in \Yscr(\pie)$ into three classes:  
\begin{enumerate}[label=(\Alph*)]
\item $y_i\in \Yscr'(\pi^*)$: For every $y_i$ of this class $\pie(y_i|i)>0$ and $\pie(y_i|x)=0,\forall x\in \Xscr\backslash\{i\}$. As a consequence \eqref{eq:proving_unique} holds for this class.
\item $y_i\in \Yscr(\pi^*)$ and $Q(i)=\phi$:
Clearly from \eqref{eq:ineq_equality}, for such a $y_i$, \begin{equation}\label{eq:ds}
\pi^*(y_i|i)>\pi^*(y_i|x),\forall x\neq i.
\end{equation} Observe that $\forall x\neq i,$ \begin{align*}
\pie(y_i|i)& \buildrel{(a)}\over\geq \pi^*(y_i|i)\buildrel{(b)}\over>\pi^*(y_i|x)\buildrel{(c)}\over\geq \pie(y_i|x),
\end{align*}
where $(a)$ follows from \eqref{eq:construc_conseq2}, $(b)$ follows from \eqref{eq:ds} and $(c)$ follows from \eqref{eq:construc_conseq}. Hence, \eqref{eq:proving_unique} holds.
\item $y_i\in \Yscr(\pi^*)$ and $Q(i)\neq\phi$: For every $y_i$ of this class, we have $\forall x\neq i$,
\begin{align*}\pie(y_i|i)&\buildrel{(a)}\over=\pi^*(y_i|i)+|Z(i)|\frac{\delta}{k}\\
&\buildrel{(b)}\over\geq \pi^*(y_i|x)+|Z(i)|\frac{\delta}{k}\buildrel{(c)}\over> \pie(y_i|x).
\end{align*}
\end{enumerate}
Here $(a)$ follows from $\eqref{eq:gh}$ and $(b)$ follows from \eqref{eq:construc_conseq}. Finally $(c)$ follows since $Z(i)$ is non-empty.
Therefore, we can conclude that \eqref{eq:proving_unique} holds. Thus, for all $k$, $\brr(\pie)=\{\sigma_k\}$, where $\sigma_k(i|y_i)=1, \forall y_i \in \Yscr(\pie)$. This shows that $\pie \in \Ubr, \forall k\in \Nbb.$ 
 Consequently, for every $k \in \Nbb$,
\begin{align*}
\underline{U}(\pie)=U(\pie,\sigma_{k})&=\sum \limits_{x\in \Xscr} \sum \limits_{y_i \in \Yscr(\pie)} \pie (y_i|x)\ut(i,x).
\end{align*}
Since $\Yscr(\pie)$ is independent of $k$, we get $\lim \limits_{k\rightarrow \infty} \underline{U}(\pie)=U(\pi^*,\sigma^*)=\OPT(\lpu).$ This proves our lemma.
\end{proof}
We can now proceed to prove our main theorem, where we show that $\OPT(\lpu)=\Ut^*$.
\begin{theorem}\label{theo:main}
For any utility function $\ut$,
$\Ut ^*=\OPT(\lpu).$
\end{theorem}
\begin{proof}
To prove our result, we begin by first showing that $\Ut^*$ is upper bounded by $\OPT(\lpu)$. Fix a $\pi \in \AS$ and a $\sigma\in \brr(\pi)$. We can construct a unique $\mu \equiv (\pi,\sigma)$. Therefore, $\underline{U} (\pi) \leq V(\mu)$. Thus,
  \begin{equation}\label{eq:jj}
 \small
 \Ut ^*\leq  \OPT(\lpu).
\end{equation}

The previous lemma guarantees the existence of a sequence of strategies $\{\pie\}_{k\in \Nbb}$ whose WCEU converges to $\OPT(\lpu)$, \ie,
\begin{equation}\label{eq:hh}
\lim \limits_{k \rightarrow \infty} \underline{U} (\pie)=\OPT(\lpu).
\end{equation}  
 But from the definition of $\ut^*$, $ \underline{U}(\pi_k)\leq \ut^*.$
Using \eqref{eq:jj} and \eqref{eq:hh} we get $\ut^*=\OPT(\lpu).$
\end{proof}
Next as a corollary we prove that the sequence of strategies constructed in Lemma \ref{lemma:sequence_convergent} is a sequence of \ese. Additionally, we also show that SGV must be the \textit{best case expected utility} for a strategy of the sender.
\begin{corollary}\label{cor:main_theo}
\begin{enumerate}
\item Let $\mu^*$  be an optimal solution of $\lpu$.  Then for each $k\in \Nbb$, there exists a $\varepsilon_k$-SES $\pi_k$ such that $\varepsilon_k \xrightarrow{k \rightarrow \infty} 0$ and $\{\pi_k\}_{k\in \Nbb}\xrightarrow{k \rightarrow \infty} \pi$, where $\mu^*\equiv (\pi,\sigma)$ for some $\sigma \in D(\pi)$.
\item Given any $\ut$, \begin{equation}\label{eq:cor-2}
\max \limits_{\pi\in \AS}\overline{U}(\pi)=\sup \limits_{\pi\in \AS}\underline{U}(\pi)=\Ut^*.
\end{equation}
\end{enumerate}
\end{corollary}
\begin{proof} For any optimal solution $\mu^*$ of $\lpu$, consider the sequence $\{\pi_k\}$ and $\sigma^*\in D(\pi^*)$ constructed in Lemma \ref{lemma:sequence_convergent}, where $\pi^*=\lim \limits_{k\rightarrow \infty} \pi_k.$ Let $\varepsilon_k =\ut^*-\underline{U}(\pi_k)$. From Theorem \ref{theo:main} it is clear that $\varepsilon_k \rightarrow 0$ as $k\rightarrow \infty$. 
 This proves part $1$ of the corollary.

From Theorem \ref{theo:main}, it is evident that $V(\mu^*)=\Ut^*$, where $\mu^*$ solves $\lpu.$ If $\max \limits_{\pi\in \AS}\overline{U}(\pi)>\Ut^*,$ then there exists a pair of $\pi'$ and $\sigma'\in \brr(\pi')$  such that $U(\pi',\sigma')>\Ut^*$. Thus, there exists a $\mu'\equiv (\pi',\sigma')$. Accordingly, $\Ut^*=V(\mu^*)<V(\mu')$ which is a contradiction. This proves \eqref{eq:cor-2}.
\end{proof}
Eq \eqref{eq:cor-2} gives a robustness to our conclusions -- they hold under both the pessimistic (max-min) Stackelberg equilibrium and also under the optimistic (max-max) one.

\subsubsection{An illustrative example}
Theorem \ref{theo:main} provides us an alternative way for computing the SGV. A brute-force approach would require us to compute the worst case expected utility for every strategy of the sender and find the supremum over all such values to arrive at the SGV. Solving the linear program $\lpu$ is instead significantly simpler. In the following example, we illustrate this.\\
\begin{examp}\label{example1}
Let $\Xscr=\{1,2\}$ and consider a utility function $\ut_1:\Xscr\times \Xscr \rightarrow \Rbb$, $\ut_1=\begin{bmatrix}
0 & 1 \\
-1 & 0 \\
\end{bmatrix}$. Here each $(i,j)^{th}$ entry represents $\ut_1(i,j)$. Let $\Yscr=\{y_1,y_2\}$.
We will categorize the strategies in $\AS$ into two different classes based on the size of the signal space.
\begin{mycases}
    \item{$|\Yscr(\pi)|=1$:}
      Let $\Yscr(\pi)=\{y\}$ WLOG. Let $C_1$ denote the collection of $\pi$'s in this class. If $\pi \in C_1$ then $\pi(y|x)=1, \forall x\in \Xscr.$ Thus, $\brr(\pi)=\AR$. But \begin{align*}
      U(\pi,\sigma)& =\pi(y|1)\sigma(2|y)\ut(2,1)+\pi(y|2)\sigma(1|y)\ut(1,2)\\
      &=-\sigma(2|y)+\sigma(1|y).
      \end{align*}
      Therefore, $\underline{U}(\pi)=-1$. Thus, $\sup \limits_{\pi \in C_1 } \underline{U}(\pi)=-1.$

     \item{$|\Yscr(\pi)|=2$:} Let $C_2$ denote the collection of all strategies in this class. Let $\pi\in C_2$ and $\pi(y_1|1)=p$ and $\pi(y_2|2)=q$, where $p,q\in [0,1]$. The strategies in $C_2$ can be further categorized into distinct classes based on the values of $p$ and $q$.
     \end{mycases}
      \begin{enumerate}[label=(\Alph*)]
\item \textbf{$p+q=1, p,q \in (0,1)$}: Let $A$ denote the collection of all strategies in this class. Clearly for a $\pi \in A$, there exists a $\sigma \in \brr(\pi)$, where $\sigma(2|y_1)=\sigma(2|y_2)=1$. Hence $\underline{U}(\pi)=U(\pi,\sigma)=-p$. And therefore $\sup \limits_{\pi \in A } \underline{U}(\pi)=-1.$
\item  \textbf{$p+q>1,p,q \in (0,1)$}: Let $B$ denote the collection of all strategies of this class. The strict inequality ensures that every $\pi \in B$ must belong to $\Ubr$ and $\brr(\pi)=\{\sigma\},$ where $\sigma(1|y_1)=1$ and $\sigma(2|y_2)=1$. Hence, $U(\pi,\sigma)=\underline{U}(\pi)=(1-q) +(1-p)(-1)= p-q.$ Therefore,  $\sup \limits_{\pi \in B } \underline{U}(\pi)=1.$ Note that this supremum is not attained in class $B.$

\item \textbf{$p+q<1, p,q \in (0,1)$}: Proceeding in a similar manner as above, we get $\pi\in \AS^*$ and $\underline{U}(\pi)=q-p$ which gives us $\sup \limits_{\pi \in C } \underline{U}(\pi)=1.$ Note that this supremum is not attained in class $C.$
\item \textbf{$p=1,q=1$}: Clearly every symbol will get recovered correctly in the unique best response strategy. Therefore, $\sup \limits_{\pi \in D } \underline{U}(\pi)=0.$
\item \textbf{$p=1,q\in [0,1)$}: Since $|\Yscr(\pi)|=2$, we must have $q>0$ which makes every $\pi\in \Ubr$. For a $\pi$ in this class, we have $\underline{U}(\pi)=1-q$. Thus, $\sup \limits_{\pi \in E} \underline{U}(\pi)=1.$ Note that this supremum is not attained in class $E.$
\item \textbf{$p\in [0,1),q=1$}: Since $|\Yscr(\pi)|=2$, we must have $p>0$ which makes $\pi \in \Ubr$. If $\pi$ is from this class then $\underline{U}(\pi)=p-1$. Thus, $\sup \limits_{\pi \in F} \underline{U}(\pi)=0.$
\item \textbf{$p\in [0,1),q\in[0,1)$}: Notice that every $\pi\in G$ belongs to $\Ubr$. Further, we have
\begin{enumerate}
\item \textbf{$p=0,q\in(0,1)$}:  $\underline{U}(\pi)=q$, which gives us $\sup \limits_{\pi \in G(a)} \underline{U}(\pi)=1.$ Note that this supremum is not attained in class $G(a).$
\item \textbf{$p=0,q=0$}: $\underline{U}(\pi)=0$, which gives us $\sup \limits_{\pi \in G(b)} \underline{U}(\pi)=0.$

\item \textbf{$p\in(0,1),q=0$}: $\underline{U}(\pi)=-p$, which gives us $\sup \limits_{\pi \in G(c)} \underline{U}(\pi)=0.$
\end{enumerate}

\end{enumerate}
Therefore, we can conclude that $\ut_1 ^*=1$. Clearly, if a SES exists then it must belong to one of the four classes in $\{B,C,E,G(a)\}$. But we have seen that the supremum is not attained in any of these classes. Therefore, a SES does not exist for $\ut_1$.

Next we compute $\ut_1^*$ using $\textbf{P}(\ut_1)$. Let $r=\mu(1|1)$ and $s=\mu(2|2)$. Notice that by plugging in the values of utility in $\textbf{P}(\ut_1)$, the linear program is transformed to the following LP, where
\begin{equation}
\begin{gathered}
 \mathbf{P}(\ut_1):\max_{(r,s)}  \quad \quad \quad \quad \quad \quad   s-r  \quad \quad \\
 \quad \quad \quad \quad \text{s.t.} r\geq 1-s,0\leq r,s\leq 1.
\end{gathered}
\end{equation}
It is easy to see that $\OPT(\textbf{P}(\ut_1))=1 $ which is attained  under the unique solution $(\mu^*(1|1)=r=1,\mu^*(2|1)=1-r=0,\mu^*(1|2)=1-s=1,\mu^*(2|2)=s=0)$. Therefore, $\OPT(\textbf{P}(\ut_1))=1=\ut_1 ^*.$
Clearly, the LP gives the $SGV$ far more easily compared to the above calculations.

Another benefit of solving the LP is that using the optimal solution of the LP we can construct a \ese \quad if no strategy of the sender attains the SGV (like this example). By glancing at our $\mu^*$, we know that the sender can benefit by making the receiver recover all symbols as $1$. Using $\mu^*$, we construct a sequence of \ese, $\{\pie\}_{k}$, with $\varepsilon=\frac{0.1}{k}$ for every $\pie$, where $\Yscr(\pie)=\{y_1,y_2\}$ and $\pie(y_1|1)=1,\pie(y_1|2)=1-\frac{0.1}{k}, \pie(y_2|2)=\frac{0.1}{k}$. Notice that $\pie$ belongs to class $E$ for all $k$ and thus $\pie\in \Ubr,\forall k$. Hence, $\underline{U}(\pie)=1-\frac{0.1}{k}$, implying $\lim \limits_{k \rightarrow \infty} \underline{U}(\pie)\rightarrow 1=\ut_1^*.$ Similar \ese \ strategies can be constructed in classes $B,C$ and $G(a)$ above.
\end{examp}
\subsubsection{Interpretation of $\lpu$ as an assignment problem}
The linear program $\lpu$ can also be viewed as an assignment problem~\cite{burkard2012assignment} with additional constraints. 
Let $\Xscr$ denote the set of $q$ tasks and also a set of $q$ agents, where every agent is allocated $q$ tasks across its working hours. Let $\mu$ denote a probability distribution where $\mu(\bullet|x)$ is the portion of working hours of agent $x$ being allocated to perform the task $\bullet \in \Xscr$. 
Another constraint imposed on a feasible allocation policy is that the portion of working hours of agent $x$ allocated for performing the task $\xhat$ can never exceed the portion of working hours of agent $\xhat $ allocated on the task $\xhat$. This is precisely the trust constraint. A reward value of $\mu(\xhat|x)\ut(\xhat,x)$ is obtained if $x$ allocates $\mu(\xhat|x)$ towards $\xhat$, where $\ut:\Xscr\times \Xscr \longrightarrow \Rbb$. Thus, for any $\ut$, the objective of this assignment problem is to find an allocation policy $\mu \in \Pscr(\Xscr|\Xscr)$ for every $x\in \Xscr$ which maximizes the expected reward, \ie,  maximizes $\frac{1}{q}\sum \limits_{x,\xhat\in \Xscr}\mu(\xhat|x)\ut(\xhat,x).$
We would like to emphasize that we are unaware of any assignment problem of this structure studied in the literature.
\section{Informativeness of the sender}\label{sec4}
This section is dedicated to quantifying the minimum amount of information revealed by the sender at equilibrium. We call this quantity the \textit{informativeness} of the utility function of the sender. We show that informativeness can be characterized by a linear program and that loss of information is imminent if there is a misalignment of interest between the players for at least one symbol. Additionally, we also show that informativeness has a linear relationship with the SGV for structured utility functions.

In our previous work \cite{deori2022information}, we studied a similar setting where the players were constrained to play deterministic strategies.
We defined informativeness denoted by $\Iscr(\ut)$ as the minimum the number of symbols that are correctly recovered by the receiver in any equilibrium. We showed that $\Iscr(\ut)$ is given the \textit{vertex clique cover number} of the strong sender graph graph defined below. \begin{definition}(Strong sender graph)
$G_s(\ut)=(\Xscr,E)$ is the strong sender graph of a utility function of the sender $\ut$ where $(x,x')\in E$ if $\ut(x,x')\geq 0$ and $\ut(x'x)\geq 0$.
\end{definition}
Thus, from \cite[Theorem 4.1]{deori2022information}, we have \begin{equation}
\Iscr(\ut)=\theta_v(G_s(\ut)),
\end{equation}
where $\theta_v(G)$ is the vertex clique cover number of any graph $G$.

We define informativeness for our setting in the following way.
\begin{definition}(Informativeness of the utility function) Let $\{\vpie\}_{\varepsilon}$ denote a sequence of \ese \quad  such that $\lim \limits_{\varepsilon \rightarrow 0} \underline{U}(\vpie)=\Ut^*$. The informativeness of a utility function of the sender $\ut$, denoted by $\Ninfo(\ut)$ is defined as:
\begin{equation}\label{eq:info}
\begin{aligned}
\Ninfo(\ut):= \inf \limits_{\{\vpie\}_{\varepsilon}} & \liminf_{\varepsilon \rightarrow 0} \Rscr(\vpie, \sigma_{\varepsilon}),\text{where}\sigma_{\varepsilon} \in \brr(\vpie). \text{\hspace{-3mm}\hspace{-3mm}}
\end{aligned}
\end{equation}
\end{definition}
First, we explain our choice of definition for measuring the minimum information recovered or revealed at equilibrium. Example \ref{example1} illustrates a utility function where no strategy of the sender attained the SGV. But as demonstrated in Theorem \ref{theo:main}, there exist a sequence of \ese,\ whose worst case expected utility approaches the SGV as $\varepsilon \rightarrow 0.$ Thus, it is sensible to measure the minimum information revealed at equilibrium by studying such sequences. Now there could exist a divergent sequence of \ese\ $\{\vpie\}$, for which $\lim \limits_{\varepsilon \rightarrow 0} \underline{U}(\vpie)=\Ut^*$. 
As a consequence $\{\Rscr(\pi_{\varepsilon},\sigma_{\varepsilon})\}$ could also be divergent. But since these sequences,  $\{\pi_{\varepsilon}\}$ and $\{\Rscr(\pi_{\varepsilon},\sigma_{\varepsilon})\}$ are bounded, they have (possibly multiple) accumulation points. Consequently, we use the $\liminf$ to capture the smallest value of the accumulation points of $\{\Rscr(\vpie,\sigma_{\varepsilon})\}$ as $\varepsilon \rightarrow 0.$ Additionally note that $\Rscr(\vpie,\sigma_{\varepsilon})$ is constant for all $\sigma_{\varepsilon} \in \brr(\vpie)$, hence the choice of the specific $\sigma_{\varepsilon}\in \brr(\vpie)$ does not matter. 
Therefore, the infimum of $\liminf \limits_{\varepsilon \rightarrow 0} \Rscr(\vpie, \sigma_{\varepsilon})$ when taken over all such sequences $\{\vpie\}$ of \ese\ strategies quantifies the minimum information correctly recovered at equilibrium.

 In this section we show that every accumulation point of a sequence of \ese \quad is equivalent to some optimal solution of $\lpu$. As a consequence, we show that informativeness is given by $\min \limits_{\mu^* } \sum \limits_{x\in \Xscr}\mu^*(x|x)$, where the minimum is over all optimal solutions $\mu^*$ of $\lpu.$ This implies that the informativeness of a utility function can also be computed by an LP. We prove these claims by first proving a few additional results on converging \ese \ strategies.
\begin{lemma}
Let $\{\vpie\}_{\varepsilon>0}$ be a sequence of \ese. If $\{\vpie\} \xrightarrow{\varepsilon \rightarrow 0}\pi$, then there exists a subsequence $\{\vpie\}_{\varepsilon \in S}$ such that $\Yscr(\pi)\subseteq \Yscr(\vpie), \forall \varepsilon \in S$.
\end{lemma}
\begin{proof} We prove by contradiction. Suppose there is a sequence $S$ of $\varepsilon \rightarrow 0$ such that $\Yscr(\pi)\nsubseteq \Yscr(\vpie), \forall \varepsilon \in S$. Then there exists a $y\in \Yscr(\pi)$ such that $y \notin \Yscr(\vpie), \forall \varepsilon \in S$. Since $y \in \Yscr(\pi)$, there exists a $x\in \Xscr$ such that $\pi(y|x)>0.$ Additionally, note that since $y \notin \Yscr(\vpie),$ therefore $\vpie(y|x)=0, \forall \varepsilon \in S$. But $\lim \limits_{\varepsilon \rightarrow 0} \vpie(y|x)=0\neq \pi(y|x)$, which is a contradiction since $\{\vpie(y|x)\}\xrightarrow{\varepsilon \rightarrow 0}\pi(y|x)$. 
\end{proof}
Given any pair of $\pi \in \AS$ and $\sigma \in \brr(\pi)$, let \begin{equation}
K(\pi,\sigma):=\{(x,y) \in \Xscr \times \Yscr| y\in \Yscr(\pi), \sigma(x|y)=1\}.
\end{equation}
Clearly if $\sigma \in \overline{D}(\pi)$, then $K(\pi,\sigma)$ is never empty.
\begin{proposition}
Let $\{\vpie\}_{\varepsilon>0}$ be a sequence of \ese \ such that $\{\vpie\}\xrightarrow{\varepsilon \rightarrow 0}\pi$. If $\sigma \in \overline{D}(\pi)$ then there exists a  subsequence $\{\sigma_{\varepsilon}\}_{\varepsilon \in S}$ such that $\sigma_{\varepsilon} \in \overline{D}(\vpie)$ and $K(\pi,\sigma) \subseteq K(\vpie,\sigma_{\varepsilon}),\forall \varepsilon \in S$.
\end{proposition}
\begin{proof} We prove this by contradiction. First note that $\overline{D}(\pi)$ is non-empty for all $\pi\in \AS$ by Lemma \ref{Lemma:determisiticBR}. Suppose there is a sequence $S$ of $\varepsilon \rightarrow 0$ such that \begin{equation}\label{eq:a1}
K(\pi,\sigma) \nsubseteq K(\vpie,\sigma_{\varepsilon}), \quad \forall \sigma_{\varepsilon}\in \overline{D}(\vpie),\forall \varepsilon \in S.
\end{equation} In the previous lemma, we proved the existence of a subsequence $\{\vpie\}_{\varepsilon \in S'}$ such that $\Yscr(\pi)\subseteq \Yscr(\vpie), \forall \varepsilon \in S'$. Therefore, there exists a  $y\in \Yscr$ and a pair of $x,x_{\varepsilon}'\in \Xscr,\forall \varepsilon \in S'$ such that $(x,y) \in K(\pi, \sigma)$ and $(x_{\varepsilon}',y) \in K(\vpie,\sigma_{\varepsilon}),\forall \varepsilon \in S'$. 
Since $x$ and $x_{\varepsilon}'$ are  elements of $\Xscr$ which is finite in size, therefore there exists a $\{\vpie\}_{\varepsilon \in S''}$, where $S''\subseteq S'$ and a fixed $x'\neq x$ such that
\begin{equation}\label{eq:a2}
\begin{split}
(x',y)\in K(\pi_{\varepsilon},\sigma_{\varepsilon}),
(x,y)\in K(\pi,\sigma),
 \forall \varepsilon \in S''.
\end{split}
\end{equation} 

 Observe that \eqref{eq:a1}
 and \eqref{eq:a2} hold if and only  if \begin{equation}\label{eq:assumption_conse}
\begin{split}
\pi(y|x)>\pi(y|x') \text{and}\vpie (y|x')> \vpie (y|x), \forall \varepsilon \in S''.
\end{split}
\end{equation} 
But note that for the sequence $\{\vpie\}_{\varepsilon \in S''}, \lim \limits_{\varepsilon \rightarrow 0} \vpie (y|x') \geq \lim \limits_{\varepsilon \rightarrow 0} \vpie (y|x) \implies \pi(y|x') \geq \pi(y|x)$,
which is a contradiction to \eqref{eq:assumption_conse}. This proves our proposition.
\end{proof}
In the following proposition, we show that if  $\varepsilon \rightarrow 0$ then any convergent sequence of \ese\ as  $\varepsilon \rightarrow 0$ must converge to a strategy which is equivalent to an optimal solution of $\lpu$.
\begin{theorem}\label{theo:mu_pi}
For any $\ut$, let $\vpie$ be an \ese \ strategy, where $\varepsilon>0$. If $\{\vpie\}\xrightarrow{\varepsilon \rightarrow 0}\pi$ then there exists a $\mu^* \equiv (\pi,\sigma)$ where $\mu^*$ is an optimal solution of $\lpu$.
\end{theorem}
\begin{proof}
The previous proposition guarantees that for any $\sigma \in \overline{D}(\pi)$, there exists a subsequence $\{\vpie\}_{\varepsilon \in S}$ such that $K(\pi, \sigma) \subseteq K(\pi_{\varepsilon},\sigma_{\varepsilon})$, where $\sigma_{\varepsilon} \in \overline{D}(\vpie), \forall \varepsilon \in S$.  Thus, $\sigma_{\varepsilon} (\bullet|y)=\sigma(\bullet|y),\forall y\in \Yscr(\pi),\forall \varepsilon \in S$ and $\lim \limits_{\varepsilon \rightarrow 0} \vpie(y|x) =\pi(y|x),\forall x \in \Xscr,\forall y\in \Yscr$. Also note that $\pi(y|x)=0,\forall x \in \Xscr,\forall y\in \Yscr'$, where $\Yscr'= \Yscr \backslash \Yscr(\pi)$. 
Accordingly, writing $\Yscr=\Yscr(\pi)\cup \Yscr'$, we get
\begin{align*}
\lim \limits_{\varepsilon \rightarrow 0} U(\vpie,\sigma_{\varepsilon})
&=\sum \limits_{x\in \Xscr}\sum \limits_{\xhat \in \Xscr} \sum \limits_{y \in \Yscr(\pi)}\pi(y|x) \sigma(\xhat|y)\ut(\xhat,x)\\
&=V(\mu^*),
\end{align*}
where $\mu^*\equiv(\pi,\sigma)$. Since $\sigma_{\varepsilon} \in \overline{D}(\vpie)$, we have $\underline{U} (\pi_{\varepsilon}) \leq U(\pi_{\varepsilon},\sigma_{\varepsilon}) \leq \Ut^*$ which follows from Corollary \ref{cor:main_theo}. Thus, letting $\varepsilon \rightarrow 0$, we get $V(\mu^*)=\Ut^*$. Therefore, we can conclude that $\mu^*$ is  an optimal solution of $\lpu$.
\end{proof}

Next, we will show that the informativeness of the utility function of the sender is exactly equal to $\OPT(\mathbf{I}(\ut))$ where $\mathbf{I}(\ut)$ is  defined as the following optimization problem,
\begin{equation}
\Ibf(\ut): \quad \min \limits_{\mu^* \in O^*} \sum \limits_{x\in \Xscr} \mu^*(x|x),
\end{equation}
where $O^*:=\{\mu^*|\mu^* \text{is an optimal solution of} \lpu\}$. $\Ibf(\ut)$ can also be expressed as an LP using LP duality. Let $\Cscr$ denote  the collection of all $(\mu,w,v)$ which satisfy the constraints of $\lpu$ and $\dpu$; and the additional constraint: $\sum \limits_{x \in \Xscr} w(x)=V(\mu).$ Thus, $\Ibf(\ut)$ can be expressed as
\begin{equation}
\begin{aligned}
\Ibf(\ut):\min_{(\mu,w,v)} \quad & \sum \limits_{x}\mu(x|x)\\
\textrm{s.t.} \quad & (\mu,w,v) \in \Cscr,
\end{aligned}
\end{equation}
which is a linear program.
\begin{theorem}\label{theorem:LandI}
For any $\ut$, $\Ninfo(\ut)=\OPT(\Ibf(\ut)).$
\end{theorem}
\begin{proof}
Since every sequence of \ese \ strategy $\{\vpie\}$ is bounded, thus every convergent subsequence $\{\vpie\}_{\varepsilon\in S}$ must converge to an accumulation point as $\varepsilon \rightarrow 0$. If $\pi$ is an accumulation point of $\{\vpie\},$ then there exists a subsequence of $\{\vpie\}\rightarrow$ which converges to $\pi$. Recall Theorem \ref{theo:mu_pi}, which guarantees the existence of an optimal solution $\mu^*$ of $\lpu$ such that $\mu^* \equiv (\pi,\sigma)$, where  $\sigma \in \overline{D}(\pi)$. Now notice that $\sum \limits_{x \in \Xscr} \mu^*(x|x)=\sum \limits_{y \in \Yscr} \sum \limits_{x \in \Xscr} \pi(y|x)\sigma(x|y)=\Rscr(\pi,\sigma)=\liminf \limits_{\varepsilon \rightarrow 0} \Rscr(\vpie,\sigma_{\varepsilon}).$ Hence, for every sequence of \ese \ $\{\vpie\}$, $\liminf \limits_{\varepsilon \rightarrow 0} \Rscr(\vpie,\sigma_{\varepsilon})\geq \OPT(\Ibf(\ut))$. In Corollary \ref{cor:main_theo}, we proved that for every optimal solution $\mu^*$ of $\lpu$, there exists a sequence of \ese\ strategies $\{\vpie\}$ which converges to some $\pi^*\in \AS$, where $\mu^* \equiv (\pi^*,\sigma^*),$ for some $\sigma^*\in D(\pi)$. Therefore, for any optimal $\mu^*$, we have $\sum \limits_{x\in \Xscr}\mu^*(x|x)=\Rscr(\pi^*,\sigma^*)=\lim \limits_{\varepsilon \rightarrow 0}\Rscr(\vpie,\sigma_{\varepsilon})\geq \inf \limits_{\{\vpie\}} \liminf \limits_{\varepsilon \rightarrow 0} \Rscr(\vpie,\sigma_{\varepsilon})$. This gives us $\OPT(\Ibf(\ut))\geq \inf \limits_{\{\vpie\}} \liminf \limits_{\varepsilon \rightarrow 0} \Rscr(\vpie,\sigma_{\varepsilon})$. Therefore, $\Ninfo(\ut)= \OPT(\Ibf(\ut)).$\end{proof}
\subsection{Properties of informativeness}
In our previous paper \cite{deori2022information}, we found that misalignment of interest between the players does not always guarantee loss of information. We demonstrated this using the utility function in the following example.

\begin{examp}\label{eg2}
Let $\ut_2=\begin{bmatrix}
0 & 1 & -1\\
-1 & 0 & 1\\
1 & -1 & 0\\
\end{bmatrix}$ be defined on $\Xscr=\{1,2,3\}$.
Observe that $G_s(\ut_2)$ has no edge resulting in $\Iscr(\ut_2)=3.$  This example is significant since it asserts that although there was misalignment of interest between the players but the restrictions to playing only deterministic strategies ensured that it was not optimal for the sender to hide information.

We now consider the setting introduced in this paper. It is easy to check that $\mu^*$ where $\mu^*(3|1)=\mu^*(1|2)=\mu^*(2|3)=\mu(x|x)=0.5, \forall x \in \Xscr$ is the unique optimal solution of $\textbf{P}(\ut_2).$ Therefore, $\OPT(\textbf{P}(\ut_2))=1.5$ and $\OPT(\Ibf(\ut_2))=1.5=\Ninfo(\ut_2)<\Iscr(\ut_2)=3.$
\end{examp}\\

Surprisingly, the comparison of $\Ninfo(\ut_2)$ with $\Iscr(\ut_2)$ indicates that loss of information is imminent in the behavioral setting for such scenarios. To this end, we prove that only pure alignment of objectives between the players can guarantee no loss of information at equilibrium. Additionally, we also show that loss of information at equilibrium can never be greater than $q-1$, \ie $\Ninfo(\ut)\geq 1$. 
\begin{theorem}\label{theorem:extreme_info}
\begin{enumerate}
\item For any $\ut$, $\Ninfo(\ut)\geq 1$.
\item $\Ninfo(\ut)=q$ if and only if $\ut(\xhat,x)<0, \forall \xhat\neq x\in \Xscr.$
\end{enumerate}
\end{theorem}
\begin{proof} For part $1$ observe that if $\mu$ is any feasible solution of $\lpu$, then $\mu(\xhat|\xhat)\geq \mu(\xhat|x), \forall \xhat,x \in \Xscr \implies \sum \limits_{\xhat \in \Xscr} \mu(\xhat|\xhat)\geq \sum \limits_{\xhat \in \Xscr} \mu(\xhat|x)=1, \forall x \in \Xscr$. Hence, $\Ninfo(\ut)\geq 1$.

Next we prove part $2$. Since $\ut(\xhat,x)<0, \forall \xhat\neq x\in \Xscr$, it follows from Theorem \ref{prop:positive_prob_positive_values} that under every optimal solution $\mu^*$ of $\lpu$, if $\mu^*(\xhat|x)>0$ then $\xhat=x.$ Thus, $\mu^*(x|x)\equiv 1$. Therefore, SGV$=0$ and $\Ninfo(\ut)=\sum \limits_{x \in \Xscr} \mu^* (x|x)=q$.

Conversely, if $\Ninfo(\ut)=q$, then $\mu^*$ such that $\mu^*(x|x)=1, \forall x \in \Xscr$ is the unique solution of solution of $\Ibf(\ut)$. But every feasible solution of $\Ibf(\ut)$ is also a solution of $\lpu$ and $\mu^*$ is the only feasible solution which gives us $\sum \limits_{x\in \Xscr}\mu^*(x|x)=q$. Thus, $\mu^*$ is also a unique solution of $\lpu$ which makes SGV$=0$. Next we will show that  $\ut(\xhat,x)<0,\forall \xhat\neq x\in \Xscr$. Suppose $\ut(\xhat,x')\geq 0$ for some $\xhat \neq x'$. There exists a feasible $\mu',$ with  $\mu'(\xhat|x')=1 $ and $\mu'(x|x)=1,\forall x \in \Xscr\backslash \{x'\}$ giving $V(\mu')=\ut(\xhat,x').$ If $\ut(\xhat,x')>0$, this contradicts that SGV $=0$.
If $\ut(\xhat,x')=0$, then $\mu'$ is also an optimal solution of $\lpu$ which is a contradiction since an optimal solution must be unique. Therefore, $\Ninfo(\ut)=q$ implies that $\ut(\xhat,x)<0, \forall x\neq \xhat \in \Xscr$. This proves our theorem.
\end{proof}

 Next, we bound the informativeness with the SGV for a particular class of utility functions. We show that if a sender's utility function is such that, for every symbol, it is not indifferent between correct recovery and incorrect recovery of the symbol, then we can bound the informativeness using SGV and vice-versa.

\begin{theorem}\label{Prop:unform_SGV_Info_rel} Let $\ut$ be such that for every pair of distinct $x$ and $\xhat$ in $\Xscr$, we have $\ut(\xhat,x)\neq 0$.  Let $A(\ut):=\{(\xhat,x)\in \Xscr \times \Xscr:\ut(\xhat,x)> 0\}$, $u^+=\max \limits_{(\xhat,x)\in A(\ut)} \ut(\xhat,x)$ and $
 u^-=\min \limits_{(\xhat,x)\in A(\ut)} \ut(\xhat,x).$ Then \begin{equation}\label{eq:sgv_info}
q-\frac{\ut^*}{u^+}\geq \Ninfo(\ut)\geq q-\frac{\ut^*}{u^-}.
\end{equation}
Additionally, if $u^+=u^-=u$, then \begin{equation}\label{uni_sgv_info}
\Ninfo(\ut)=q-\frac{\ut^*}{u}.
\end{equation}
\end{theorem}
\begin{proof}
Consider an optimal solution $\mu^*$ of $\lpu$ such that $\Ninfo(\ut)=\sum \limits_{x\in \Xscr} \mu^*(x|x)$.
Now observe that $\ut^*=V(\mu^*)=\sum \limits_{x\in \Xscr}\sum \limits_{\xhat\neq x \in \Xscr}\mu^*(\xhat|x)\ut(\xhat,x)\geq \sum \limits_{x\in \Xscr}(1-\mu^*(x|x))u^-
\geq u^-(q-\Ninfo(\ut))$. Similarly, $u^+(q-\Ninfo(\ut))\geq \ut^*$, resulting in \eqref{eq:sgv_info}. And \eqref{uni_sgv_info} follows immediately from \eqref{eq:sgv_info} if $u^+=u^-$.
\end{proof}
Clearly $u^+=u^-$ implies that $q-\sum_{x\in \Xscr}\mu^*(x|x)$ is constant for all optimal solutions $\mu^*$ of $\lpu.$ First recall Corollary \ref{cor:main_theo}, which guarantees the existence of a convergent sequence of \ese \ approaching the limit $\pi$ as $\varepsilon \rightarrow 0$, where $\mu^*\equiv (\pi,\sigma)$ for some $\sigma \in \brr(\pi)$. Additionally,  Theorem \ref{theo:mu_pi} proves that every accumulation point of a sequence of \ese \ with $\varepsilon \rightarrow 0$ is equivalent to an optimal solution of $\lpu$.  Hence, the expected number of symbols incorrectly recovered in these accumulation points is the same for each point and exactly equal to $q-\Ninfo(\ut)$.  Therefore, the expected number of symbols incorrectly recovered in any \ese \ approaches $q-\Ninfo(\ut)$ as $\varepsilon \rightarrow 0$.
\section{Graph theoretic characterization}\label{sec5}
We introduce the notion of an \textit{obfuscation graph} in this section which helps us identify the symbols that can be grouped together in a persuasion policy. For chains, cycles and stars, we characterize $\ut^*$ and $\Ninfo(\ut)$ using this graph.
\begin{definition}
Given any $\ut$, the \textit{obfuscation graph} of $\ut$ denoted by $\Gscr(\ut)=(\Xscr,E)$ is a directed graph where $(x,x')\in E$ is a directed edge from $x$ to $x'$ if $\ut(x',x)\geq 0$. The utility $\ut(x',x)$ associated with the edge $(x,x')$ is called the \textit{weight of the directed edge} $(x,x')$ in $\Gscr.$
\end{definition}
Let $\Xscr=\{x_1,\hdots,x_q\}$. We call $(\Xscr,E)$ a \textit{chain graph} $P_q$ if $E=\{(x_i,x_{i+1})|i=1,\hdots,q-1\}$ and call it a \textit{cycle graph} denoted by $C_q=(\Xscr,E)$ if $E=\{(x_i,x_{i+1})|i=1,\hdots,q-1\}\cup\{(x_{k},x_1)\}$. Let $S_{x'}=(\Xscr,E)$ denote a \textit{directed star} with respect to node $x' \in \Xscr$ if $E=\{(x,x')| x\in \Xscr \}$.
Let $M$ denote any \textit{matching} of $\Gscr$ and $W(M)$ denote the weight of the matching $M$, where $W(M)=\sum \limits_{(x,y) \in M} \ut(y,x)$. Let $\nu(\Gscr)$ denote the \textit{weight of the matching} in $\Gscr$ with the \textit{largest weight} and let $\Mscr(\Gscr)$ denote a \textit{maximum weighted matching} where $W(\Mscr(\Gscr))=\nu(\Gscr)$. We define $\Wscr(\Gscr):=\{M|W(M)=\nu(\Gscr)\}$ and let $\underline{\nu}(\Gscr)=\min \limits_{M\in \Wscr(\Gscr)} |M|.$ For any  matching $M$ of a graph $\Gscr=(\Xscr,E)$, let $\Xscr_M$ be the set of vertices covered by $M$.

First, we prove that the SGV of any utility function is lower bound by the maximum possible weight of a matching in the corresponding obfuscation graph. 
\begin{proposition}\label{prop:matching_feasible}
 For any $\ut$,
$\nu(\Gscr(\ut))\leq \ut^*.$
\end{proposition}
\begin{proof}
Let $M$ be a matching in $\Gscr(\ut)$. We will construct a $\mu\in \Pscr(\Xscr|\Xscr)$ which will give us $V(\mu)=W(M).$ Let $\mu(x_i|x_j)=\mu(x_i|x_i)=1$ if $(x_j,x_i)\in M$ and $\mu(x_i|x_i)=1$ if $x_i\in \Xscr\backslash \Xscr_M.$ This structure ensures that $\mu\in \Pscr(\Xscr|\Xscr)$ satisfies the trust constraints. Notice that $V(\mu)=\sum \limits_{(x_i,x_j)\in M}\ut(x_j,x_i)=W(M).$ This proves our result.
\end{proof}

Thus, graph theoretic characterization gives us a lower bound for the SGV. For computing the SGV, we need to identify the maximum possible value of the function $V(\mu)$ for every feasible $\mu$ in $\lpu$. We show that this value varies as the structure of the graph changes. Accordingly, we successfully compute the SGV for stars, chains, and cycles.
\subsection{SGV for stars, chains, and cycles}
In this section, we characterize the SGV for those utility functions whose obfuscation graphs are stars, chains, and cycles. First, we prove that if the obfuscation graph of a utility function is a star, then the SGV is equal to the sum of the weight of all the edges in the graph. For a chain $P_q$, we show that the SGV is $\nu(P_q),$ while in a cycle, the SGV is the maximum of $\nu(C_q)$ and half the sum of the weight of all the edges in the cycle.
\begin{proposition}(SGV for a star)\label{star_sgv}
If $\Gscr(\ut)=S_{x'}$, then \begin{equation}\ut^*=\sum \limits_{x\neq x'}\ut(x',x).\end{equation}
\end{proposition}
\begin{proof}
For every optimal solution $\mu^*$ of $\lpu$, with $\Gscr(\ut)=S_{x'}$, we have \[V(\mu^*)=\sum \limits_{x\in \Xscr}\sum \limits_{\xhat\in \Xscr}\mu^*(\xhat|x)\ut(\xhat,x)\leq \sum \limits_{x\neq x'}\ut(x',x).\]
Note that a feasible solution $\mu'$ of $\lpu$, with $\mu'(x'|x)=1,\forall x\in \Xscr$, gives us $V(\mu')=\sum \limits_{x\neq x'}\ut(x',x).$ This proves our result.
\end{proof}

For any $\mu\in \Pscr(\Xscr|\Xscr)$, let
$\mu_{j|i}=\mu(x_{j}|x_i)$ and let
$u_i=\ut(x_{i+1},x_{i}), 1\leq i,j\leq q$. For ease of notation we identify $q+1\equiv 1$. Next for any chain $P_q$, let $u_q=0$. Consequently, $\lpu$ for any $\ut$ with $\Gscr(\ut)$ as a chain or a cycle is now equivalent to solving the following LP:
 \begin{align*}
       \lpu:&& \quad \max_{\mu}  & \sum \limits_{i \in \{1,\hdots,q\}}\mu_{i+1|i}u_i  \nonumber             \\
        \text{s.t.} &&\quad   \mu_{i+1|i+1} &\geq \mu_{i+1|i}, \quad \forall i\in \{1,\hdots,q\}\\
      && \mu &\in \Pscr(\Xscr|\Xscr).
    \end{align*}
    We denote the dual variables by $w_i=w(x_i)$ and $v_{i,j}=v(x_i,x_j).$ Using these variables, $\dpu$ can be written as

\begin{align*}
\dpu: && \quad \min_{w,v} & \sum \limits_{i=1}^q w_i \nonumber \\
\textrm{s.t.} && \quad  w_i-v_{i,i-1}&\geq 0, 1\leq i\leq q \\
 && w_i + v_{i+1,i}-u_i&\geq 0,1\leq i\leq q\\
 && v_{i,j}&\geq 0,\forall i\neq j, 1\leq i,j\leq q\\
 && w_i,& \text{unrestricted}, 1\leq i\leq q.
\end{align*}
  For a cycle $C_q$, for ease of notation we identify $v_{1,q}\equiv v_{q+1,q}$ and $v_{1,1-1}\equiv v_{1,q}$. In the following theorem, we prove that the SGV for a chain $P_q$ is $\nu(P_q).$
 \begin{proposition}(SGV for a chain)\label{lemma:chain_opt_value}
Let $\ut$ be a utility function defined on $\Xscr=\{x_1,\hdots, x_q \}$ such that $\Gscr(\ut)=P_q$. Then,
\begin{equation}
\ut^*=\nu(P_q).
\end{equation}
\end{proposition}
 \begin{proof}
Recall that $\ut(x,x)=0,\forall x\in \Xscr$. Therefore, $\ut^*=\sum \limits_{\xhat \neq x \in \Xscr}\mu^*(\xhat|x)\ut(\xhat,x)$, where $\mu^*$ is an optimal solution of $\lpu.$  Recall Proposition \ref{prop:positive_prob_positive_values}, where we showed that under every optimal policy $\mu^*$, if $\xhat\neq x$ then $\mu^*(\xhat|x)>0$ only if $\ut(x,\xhat)\geq 0.$ Thus, if $\Gscr(\ut)=P_q$ then $\ut^*=\sum \limits_{i=1}^{q-1}\mu^*_{i+1|i} u_i$.

Note that the trust constraints for such a $\ut$ can be equivalently written as $\mu^*_{i+1|i}+\mu^*_{i+2|i+1} \leq \mu^*_{i+1|i+1}+\mu^*_{i+2|i+1}=1, \forall i \leq q-2$. 
Consequently, $\lpu$ is now equivalent to $\max \limits_{\mu} \sum \limits_{i \in \{1,\hdots,q-1\}}\mu_{i+1|i}u_i$, where $ \mu_{i+1|i} +\mu_{i+2|i+1} \leq 1, 1\leq i \leq q-2$ and $ 0\leq \mu_{i+1|i}\leq 1,1\leq i \leq q-1.$ 
But this optimization problem is a linear relaxation of an integer program which is equivalent to finding the \textit{maximum weighted matching} in a graph. Using \cite{duan2014linear}, it is evident that an integral optimal solution must exist which is equivalent to a maximum weighted matching.
 Hence, for every feasible $\mu$, we have $\sum \limits_{i=1}^{q-1} \mu_{i+1|i} u_i \leq \nu(P_q)$. But Proposition \ref{prop:matching_feasible}, guarantees existence of a feasible $\mu'$ for which $V(\mu')=\nu(P_q).$ Consequently, $\ut^*=\nu(P_q)$.
 \end{proof}

In the following theorem, we prove that for any $\ut$, if the obfuscation graph is a cycle $C_q$, then the $\ut^*$ is the maximum of the two values: either half of the sum of the weight of all the edges in the cycle or $\nu(C_q)$.
 \begin{theorem}(SGV for a cycle)\label{theorem:Cycle:property}
  Let $\Gscr(\ut)=C_q$ be a cycle with vertex set $\Xscr=\{x_1,\hdots,x_q\}$. Then \begin{equation}
  \ut^*=  \max \{\frac{1}{2} \sum \limits_{i=1}^{q} u_i, \nu(C_q)\}.
  \end{equation}
 \end{theorem}
\begin{proof}
We prove this result by taking two types of optimal solutions of $\lpu$. If $O^*:=\{\mu^*|\mu^* \text{is an optimal solution of} \lpu\}$, then let $A:=\{\mu^*\in O^*|\mu^*_{i+1|i} >0,\forall i \in \{1,\hdots,q\}$.  
Let $\mu^*$ be an optimal solution of $\lpu$ such that $\mu^*\in A$. From complementary slackness conditions, we get $w_i=u_i- v_{i+1,i}$ and
$w_{i}=v_{i,i-1},1\leq i\leq q.$

Accordingly, $
\sum \limits_{i=1}^q w_{i}=  \sum \limits_{i=1}^{q} (u_i- v_{i+1,i})=\sum \limits_{i=1}^{q} v_{i,i-1}$. 
 As a consequence,  \begin{equation}
\sum \limits_{i=1}^q w(x_i)= \frac{1}{2} \sum \limits_{i=1}^{q} u_i.
\end{equation}
Note that we can always construct a feasible policy $\mu'$ where $\mu'_{i|i}=\mu'_{i+1|i}=0.5,$ for $1\leq i \leq $. Accordingly, $V(\mu')=\frac{1}{2}\sum \limits_{i=1}^q u_i$. Thus, $\ut^*=V(\mu')=\frac{1}{2}\sum \limits_{i=1}^q u_i$, where $\mu'\in A$.

  Consider an optimal solution $\mu^*\in A^c$. Without loss of generality suppose $\mu^*_{1|q}=0$. Thus, $\mu^*_{q|q}=1$ and $V(\mu^*)=\sum \limits_{i=1}^{q-1}\mu^*_{i+1|i}u_i$. But notice that $\mu^*$ is also an optimal solution for the chain $P_q$. Thus, $\nu(P_q)=V(\mu^*)=\ut^*$. Now every matching in $P_q$ is a matching in $C_q$. But Proposition \ref{prop:matching_feasible} implies that $\nu(C_q)\leq \ut^*.$ Accordingly, we get $\nu(C_q)=\ut^*.$ 
Therefore, we can conclude that the optimal value is $\max \{\frac{1}{2} \sum \limits_{i=1}^k u_i, \nu(C_q)\}.$
\end{proof}

If the weight of an edge is positive and uniform, then the weight of the largest matching is $\nu(P_q)$ in $P_q$ and $\nu(C_q)$ in $C_q$. Consequently, we get the following result.
\begin{corollary}\label{cor:cycle_uni}
\begin{enumerate}
\item Let $u>0$ and let $\ut$ be such that $\Gscr(\ut)=C_q$, where $u_i=u,$ for $1\leq i\leq q$.  Then, \[
\ut^*=\frac{qu}{2}.\]
\item Let $u>0$ and let $\ut$ be such that $\Gscr(\ut)=P_q$, where $u_i=u,$ for $1\leq i\leq q-1$.  Then, \[\ut^*=\left\{
\begin{array}{ll}
      \frac{(q-1)u}{2}& \text{if} q \text{is odd}\\
     \frac{qu}{2}& \text{if} q \text{is even}.
\end{array}
\right.\]
\end{enumerate}
\end{corollary}
\begin{proof}
In an odd cycle, the largest matching will have $\frac{q-1}{2}$ edges and while an even cycle will have $\frac{q}{2}$ edges in the largest matching. Therefore, $\nu(C_q)=\frac{(q-1)u}{2}$ if $q$ is odd and $\nu(C_q)=\frac{qu}{2}$ if $q$ is even. 
Therefore, $\nu(C_q)\leq \frac{qu}{2},\forall q$ and $\ut^*=\frac{qu}{2}$ follows from the previous theorem. This proves part $1$ of the corollary.

Since the weight of the edges in the chain $P_q$ are uniform, $\nu(P_q) = \frac{(q-1)u}{2}$ for $q$ odd and $\nu(P_q)= \frac{qu}{2}$ for $q$ even.
Since $\ut^*=\nu(P_q)$ and this proves part $2$ of the corollary.
\end{proof}
\subsection{Informativeness in stars, chains and cycles}
Although, we have shown that loss of information is imminent in the behavioral setting in the presence of non-negative utility values, we cannot claim that loss of information is more in the deterministic setting when compared with the behavioral setting. We prove this by bounding $\Ninfo(\ut)$ with $\Iscr(\ut)$ for different obfuscation graphs. 
\begin{proposition}
\begin{enumerate}
\item Let $\Gscr(\ut)$ be a directed and complete graph, \ie, $\ut(x,x')\geq 0,\forall x,x'\in \Xscr$. Then \[\Ninfo(\ut)\geq \Iscr(\ut)=1.\]
\item If $\Gscr(\ut)=S_{x'},x'\in \Xscr$ then $1=\Ninfo(\ut)<\Iscr(\ut)=q.$
\item Let $u>0$ and $\Gscr(\ut)=C_q$, with $u_i=u$ for $1\leq i\leq q$. Then, \[
\frac{q}{2}=\Ninfo(\ut)<\Iscr(\ut)=q.\]
\item  Let $u>0$ and $\Gscr(\ut)=P_q$, with $u_i=u$, for $1\leq i\leq q-1$. Then, \[
\Ninfo(\ut)=\left\{
\begin{array}{ll}
      \frac{(q+1)}{2}& \text{if} q \text{is odd}\\
     \frac{q}{2}& \text{if} q \text{is even}
\end{array}
\right. \]
and  $\Ninfo(\ut)<\Iscr(\ut)=q.$
\end{enumerate}
\end{proposition}
\begin{proof}\begin{enumerate}
\item Note that since $\ut(x,x')\geq 0,\forall x,x'\in \Xscr$, the sender graph $G_s(\ut)$ is a \textit{clique}. Accordingly, $\Iscr(\ut)=1\leq \Ninfo(\ut).$
\item Consider the optimal $\mu'$ constructed in the proof of Proposition \ref{star_sgv}. Observe that $\sum \limits_{x\in \Xscr}\mu'(x|x)=1$. Thus, using Theorem \ref{theorem:extreme_info}, we conclude that $\Ninfo(\ut)=1$. Next note that if $\Gscr(\ut)=S_x'$, then $G_s(\ut)$ has no edges. Therefore, $\Iscr(\ut)=q>\Ninfo(\ut).$
\item Theorem \ref{Prop:unform_SGV_Info_rel} and Corollary \ref{cor:cycle_uni} give us $\Ninfo(\ut)\geq\frac{q}{2}.$ But $\mu^*$ with $\mu^*_{i+1|i}=\mu^*_{i|i}=0.5,1\leq i\leq q$ is an optimal solution of $\lpu$ satisfying $\sum \limits_{i=1}^q\mu^*_{i|i}=\frac{q}{2}.$ Therefore, $\Ninfo(\ut)=\frac{q}{2}.$ Since $\Gscr(\ut)=C_q$, $G_s(\ut)$ has no edge. This makes $\Iscr(\ut)=q$ which proves our result.
\item Using Theorem \ref{Prop:unform_SGV_Info_rel} and Corollary \ref{cor:cycle_uni}, the result follows immediately as $\Ninfo(\ut)=q-\frac{\ut^*}{u}$.
\end{enumerate}\end{proof}
\section{Conclusion}\label{sec6}
In any persuasion setting, it is natural to ask, what makes the receiver trust the sender's suggestion? We addressed this question in our paper. We characterized the equilibrium strategies of the sender through a linear program with trust constraints. These constraints ensured that every feasible persuasion policy reveals enough information to persuade the receiver to pick a particular action. We found that revealing some true information is mandatory and quantified the minimum amount of information that needs to be revealed in any equilibrium using another linear program.
\section*{Acknowledgements}
This research was supported by the grant
CRG/2019/002975
of the Science and Engineering Research Board, Department of Science and Technology, India. The authors also acknowledge the support of the Trust Lab at IIT Bombay.
\bibliographystyle{unsrt}
\bibliography{new_ref} 
\end{document}